\documentclass[a4paper,11pt]{article}
\usepackage{amsmath,amssymb,latexsym,graphics,epsfig}
\usepackage{multicol}
\usepackage{fullpage}
\usepackage{booktabs}
\usepackage{amsmath,amssymb,latexsym,graphics,epsfig}
\usepackage{color}
\usepackage{amsthm}
\usepackage{enumerate}
\usepackage{graphicx,url}
\usepackage[hidelinks]{hyperref}
\usepackage{tikz}
\usepackage{cite}
\usetikzlibrary{arrows,shapes}

\def\1{\mathbbm{1}}

\theoremstyle{definition} 
\newtheorem{definition}{Definition}

\theoremstyle{plain} 
\newtheorem{theorem}{Theorem}
\newtheorem{corollary}{Corollary}
\newtheorem{proposition}{Proposition}
\newtheorem{que}{Question}

\title{On the complexity of solving a decision problem with flow-depending costs: the case of the IJsselmeer dikes}

\author{Aida Abiad\thanks{Corresponding author. Department of Quantitative Economics, Maastricht University, Maastricht, The Netherlands. \texttt{A.AbiadMonge@maastrichtuniversity.nl}.}
    \and Sander Gribling\thanks{CWI, Amsterdam, The Netherlands. \texttt{gribling@cwi.nl}}
    \and Domenico Lahaye\thanks{Delft Institute for Applied Mathematics, The Netherlands. \texttt{d.j.p.lahaye@tudelft.nl}}
    \and Matthias Mnich\thanks{Department of Quantitative Economics, Maastricht University, Maastricht, The Netherlands. \texttt{m.mnich@maastrichtuniversity.nl}}
    \and Guus Regts\thanks{Korteweg de Vries Institute for Mathematics, University of Amsterdam, The Netherlands. \texttt{guusregts@gmail.com}}
    \and Lluis Vena\thanks{Korteweg de Vries Institute for Mathematics, University of Amsterdam, The Netherlands. \texttt{lluis.vena@gmail.com}}
    \and Gerard Verweij\thanks{CPB Netherlands Bureau for Economic Policy Analysis, The Hague, The Netherlands. \texttt{G.Verweij@cpb.nl}}
    \and Peter Zwaneveld\thanks{CPB Netherlands Bureau for Economic Policy Analysis, The Hague, The Netherlands. \texttt{P.J.Zwaneveld@cpb.nl}}
}

\date{}

\begin{document}
\maketitle

\begin{abstract}
\noindent
  We consider a fundamental integer programming (IP) model for cost-benefit analysis flood protection through dike building in the Netherlands, due to Verweij and Zwaneveld.
  Experimental analysis with data for the Ijsselmeer lead to integral optimal solution of the linear programming relaxation of the IP model.
  This naturally led to the question of integrality of the polytope associated with the IP model.

  In this paper we first give a negative answer to this question by establishing non-integrality of the polytope.
  Second, we establish natural conditions that guarantee the linear programming relaxation of the IP model to be integral.
  We then test the most recent data on flood probabilities, damage and investment costs of the IJsselmeer for these conditions.
  Third, we show that the IP model can be solved in polynomial time when the number of dike segments, or the number of feasible barrier heights, are constant.
\end{abstract}

\noindent
\textbf{Keywords.} Cost-benefit analysis; dynamic programming; integer programming

\clearpage
\pagebreak

\section{Introduction}
\label{sec:introduction}

Protection against increasing sea levels is an important issue around the world, including the Netherlands.
Optimal dike heights are of crucial importance to the Netherlands as almost $60\%$ of its surface is under threat of flooding from sea, lakes, or rivers.
This area is protected by more than $3500$ kilometers of dunes and dikes, which require substantial yearly investments of more than one billion Euro~\cite{CPB17Stan}.

Recently, Zwaneveld and Verweij~\cite{CPB17Stan} presented an integer programming (IP) model for a cost-benefit analysis
to determine optimal dike heights that allows highly flexible input parameters for flood probabilities, damage costs and investment costs for dike heightening.
Their model improves upon an earlier model by Brekelmans et al.~\cite{BHRE2012}, who presented a dedicated approach without optimality guarantee, and which was in turn an improvement of the original model by van Dantzig~\cite{vandantzig56} from 1956.
The latter was introduced after a devastating flood in the Netherlands in 1953, with the goal of designing a long-lasting cost-efficient layout for a dike ring.

Our work is based on the integer programming model of Bos and Zwaneveld~\cite{CPB12}, Zwaneveld and Verweij~\cite{CPB14} and a recent manuscript by Zwaneveld and Verweij~\cite{CPB17}, where the authors study the problem of economically optimal flood prevention in a situation in which multiple barrier dams and dikes protect the hinterland to both sea level rise as well as peak river discharges.
Current optimal flood prevention methods (Kind~\cite{K2011}, Brekelmans et al.~\cite{BHRE2012}, Zwaneveld and Verweij~\cite{CPB17Stan}) only consider single dike ring areas with no interdependency between dikes.
Zwaneveld and Verweij~\cite{CPB14,CPB17} present a graph-based model for a cost-benefit analysis to determine optimal dike heights with multiple interdependencies between dikes and barrier dams.
Zwaneveld and Verweij~\cite{CPB14} identify several solution approaches (e.g.~a dynamic programming heuristic and branch-and cut), and they also show that it can be solved quickly to proven optimality using a branch-and-cut approach for real world problem instances.

The natural question arising from the work of Zwaneveld and Verweij~\cite{CPB17} is whether the linear programming relaxation of their IP model always admits an integral optimum.

\subsection{Our contribution}
Our first contribution is a negative answer to the question above.
In particular, we show that the polytope associated to the IP model of Zwaneveld and Verweij~\cite{CPB14,CPB17} is not necessarily integral.

Second, we derive sufficient conditions that ensure the LP relaxation to be integral.
We then experimentally verify whether these conditions are met by the most recent data on flood probabilities, damage and investment costs, which are presently used by the Dutch government.
Finally, we show that the optimal dike heightening problem can be solved in polynomial time if either the number of barrier heights or the number of dike segments is constant.

This paper is organized as follows.
In Section~\ref{sec:modeldesc} we recap the IP model of Zwaneveld and Verweij~\cite{CPB17} that forms the subject of our investigations.
In Section~\ref{sec:integrability} we discuss integrality of the polytope.
In Section~\ref{sec:altapproach} we propose an alternative approach to solve the problem by means of dynamic programming.
Finally, in Section~\ref{sec:abstraction} we present a natural abstract version of the dike height problem, which allows for several variations and open problems.

\section{Integer programming model}
\label{sec:modeldesc}
In this section we present the IP model formulated by Zwaneveld and Verweij~\cite{CPB17}.
Before going into the details of the IP model, let us introduce some important terminology and the geographical configuration of the dikes in the Netherlands.
A \emph{dike segment} is a part of a dike that is protecting a region.
It is possible that several segments protect the same area and in that case they are called a \emph{dike ring}.
In the Netherlands, dike ring areas and smaller dikes lie beneath the \emph{Afsluitdijk} (or \emph{barrier dam}) which is the outermost dike located in the north.
The Afsluitdijk separates the North Sea and the IJsselmeer, an artificial lake; see Fig.~\ref{fig:ijsselmeer} for an illustration.

\begin{figure}[htpb]
  \begin{center}
     \includegraphics[scale=0.5]{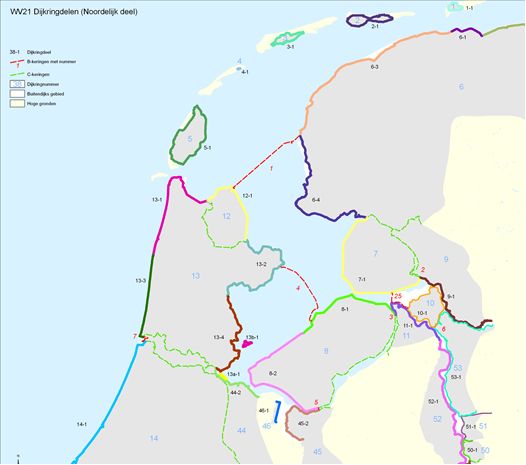}
  \end{center}
  \caption{\label{fig:ijsselmeer}The relative locations of the North Sea, the IJsselmeer, and the Afsluitdijk and the dike ring enclosing it.}
\end{figure}

The IP model uses the following data:
\begin{itemize}
  \item $D$ is the set of dike segments.
  \item $H_D$ is the set of possible heights for a dike segment. For ease of notation, we do not let $H_D$ depend on the dike segment, i.e., all dike segments have the same set of possible heights. 
  We denote the height of a previous year by $h_1$, and that of the current year by~$h_2$.
  Likewise, $H_B$ is the set of possible heights for the barrier dam and we denote the height of the barrier in the previous year by $h_1^B$, and that of the current year by $h_2^B$.
  \item $T$ is the set of time periods at which changes to a dike segment can be made (e.g., one can assume that changes are scheduled per year), for simplicity we assume (with abuse of notation) $T =  \{0,1, \ldots, T\}$.
\end{itemize}

The decision variables are:
\begin{itemize}
  \item $CY(t,d,h_1,h_2)$: this binary variable takes value 1 if dike ring $d$ is updated in time period $t$ from height $h_1$ up to height $h_2$.
    If $h_1=h_2$ then this dike ring segment is not strengthened in period $t$ and remains at its previous height.
    This decision variable is used for tracking investment (and maintenance) costs. 
  \item $DY(t,d,h_2,h_2^B)$: this binary variable takes value 1 if at the end of period $t$ the barrier dam has height $h_2^b$, and dike segment $d$ is of height $h_2$.
    This variable is used to connect investments in dike segments (and the barrier dam) to expected damages.
    Another way to view it is that this variable linearizes the $0$-$1$ variable $\left(\sum_{h_1} CY(t,d,h_1,h_2)\right) \big(\sum_{h_1^B} B(t,h_1^B,h_2^B)\big)$. 
  \item $B(t,h_1^B,h_2^B)$: this binary variable take value 1 if the barrier dam is updated in time period $t$ from height $h_1^B$ up to $h_2^B$.
    If $h_1^B=h_2^B$ then the barrier dam is not strengthened in period $t$ and remains at its previous height.
    This decision variable is used for bookkeeping investment (and maintenance) costs, flood probabilities and related expected damage costs of the barrier dam.
\end{itemize}

The input parameters are:
\begin{itemize}
  \item $D_{\text{cost}}(t,d,h_1,h_2)$, the cost for investment and maintenance, if dike ring $d$ is strengthened in time period $t$ from $h_1$ to $h_2$.
    If $h_1=h_2$, the dike ring segment is not strengthened and these costs only represent maintenance costs.
  \item $D_{\text{expdam}}(t,d,h_2,h_2^B)$, the expected damage, i.e.,
  \[
    D_{\text{expdam}}(t,d,h_2,h_2^B)= \text{prob}(t,d,h_2,h_2^B) \times \text{damage}(t,d,h_2,h_2^B),
  \]
  where $\text{prob}(t,d,h_2,h_2^B)$ and $\text{damage}(t,d,h_2,h_2^B)$ are respectively the probability of failure and the expected damage cost (the latter given that there is a flooding) in period $t$ given the height of the segment $h_2$ and the height of the barrier $h_2^B$. 
  Note that it is assumed that both the probability of failure and the expected damage upon failure of dike segment~$d$ only depend on the height of segment $d$ and that of the barrier dam.
  \item $B_{\text{cost}}(t,d,h_1^B,h_2^B)$, the cost for investment and maintenance, if the barrier dam is strengthened in time period $t$ from $h_1^B$ to $h_2^B$.
    If $h_1^B=h_2^B$, the barrier dam is not strengthened and these costs only represent maintenance costs.
  \item $B_{\text{expdam}}(t,h_2^B)$, the expected damage of a flooding of the barrier dam, i.e.~$\text{prob}(t,h_2^B)\times \text{damage}(t,h_2^B)$, here $\text{prob}(t,h_2^B)$ and $\text{damage}(t,h_2^B)$ are respectively the probability of failure and the expected damage cost (the latter given that there is a flooding), in period $t$ given the height of the barrier $h_2^B$. 
\end{itemize}

All input parameters are calculated in net present value of a certain year (i.e.~2020, which is the starting year for our calculations) and represent price levels in a certain year.

All in all, the IP model then reads as follows:
\begin{align}
  \min~ & \sum_{t\in T} \sum_{d\in D} \sum_{h_1\in H_D} \sum_{h_2\geq h_1} D_{\text{cost}}(t,d,h_1,h_2) \cdot CY(t,d,h_1,h_2) \label{eq.of_1}\\[2mm]
      + & \sum_{t\in T} \sum_{d\in D} \sum_{h_2\in H_D} \sum_{h_2^B} D_{\text{expdam}}(t,d,h_2,h_2^B) \cdot DY(t,d,h_2,h_2^B) \label{eq.of_2}\\[2mm]
      + & \sum_{t\in T} \sum_{h_1^B\in H^B} \sum_{h_2^B\geq h_1^B} \left(B_{\text{cost}}(t,h_1^B,h_2^B)+B_{\text{expdam}}(t,h_2^B)\right) \cdot B(t,h_1^B,h_2^B) \label{eq.of_3}
\end{align}
subject to
\begin{align}
\label{eq.ini} CY(0,d,0,0)=1, CY(0,d,h_1,h_2)=0 \quad &\forall d \in D, h_1,h_2\in H_D, h_2\geq h_1\wedge h_2> 0  \\[2mm]
\label{eq.222} \sum_{h_1\leq h_2} CY(t-1,d,h_1,h_2)=\sum_{h_3\geq h_2}CY(t,d,h_2,h_3) \quad &\forall t \in T_{>0},d\in D,h_2\in H_D\\[2mm]
\label{eq.3}\displaystyle\sum_{h_1\leq h_2} CY(t,d,h_1,h_2)=\displaystyle\sum_{h_2^B}DY(t,d,h_2,h_2^B) \quad &\forall t \in T,d\in D,h_2\in H_D\\[2mm]
\label{eq.ini2}B(0,0,0)=1, B(0,h_1^B,h_2^B)=0 \quad  &\forall h_1^B,h_2^B\in H_B, h_2^B\geq h_1^B\wedge h_2^B>0 \\[2mm]
\label{eq.4}\displaystyle \sum_{h_1^B\leq h_2^B} B(t-1,h_1^B,h_2^B)=\displaystyle\sum_{h_3^B\geq h_2^B}B(t,h_2^B,h_3^B) \quad &\forall t \in T\backslash \{0\},d\in D,h_2^B\in H_B\\[2mm]
\label{eq.5a}\displaystyle\sum_{h_1^B\leq h_2^B} B(t,h_1^B,h_2^B)=\displaystyle\sum_{h_2}DY(t,d,h_2,h_2^B) \quad &\forall t \in T,d\in D,h_2^B\in H_B\\[2mm]
\label{eq.ini3}CY(t,d,h_1,h_2)\in\{0,1\} \quad &\forall t\in T,d\in D, h_1\in H_D,h_2\geq h_1 \in H_D\\[2mm]
\label{eq.ini4}DY(t,d,h_2,h_2^B)\in\{0,1\} \quad &\forall t\in T,d\in D, h_2\in H_D,h_2^B\in H_B\\[2mm]
\label{eq.ini5}B(t,h_1^B,h_2^B)\in\{0,1\} \quad &\forall t\in T,d\in D, h_2^B\geq h_1^B\in H_B
\end{align}

Equations (\ref{eq.3}) and (\ref{eq.5a}) are the linking constraints between the barrier and the dike segments using the variables $DY$. Equations (\ref{eq.222}) and (\ref{eq.4}) are flow conditions. Equations (\ref{eq.ini}) and (\ref{eq.ini2}) are the initial conditions. Equations (\ref{eq.ini3}), (\ref{eq.ini4}) and (\ref{eq.ini5}) are integrality constraints.

\section{On the integrality of the polytope}
\label{sec:integrability}
The linear programming relaxation of the IP model from the previous section allows the decision variables to take values from the interval $[0,1]$ instead of the integral $\{0,1\}$.
We now give an example showing that the polytope defined by this relaxation can have vertices with non-integral coordinates.

The example involves the following sets indexing the variables:
\begin{itemize}
  \item $T=\{0,1,2\}$
  \item one segment. Hence, we remove the dike index from all related variables. 
  \item $H=\{0,1\}$, $H_B=\{0,1\}$
\end{itemize}

The point $P$, candidate to be a vertex of the polytope of the linear relaxation, has the following non-zero values:
\begin{center}
  \begin{tabular}{llllll}
    \toprule
    $(t,h_1,h_2)$     & $(0,0,0)$ & $(1,0,1)$ & $(1,0,0)$ & $(2,1,1)$ & $(2,0,0)$ \\
    \midrule
    $CY(t,h_1,h_2)$   & $1$       & $1/2$     & $1/2$     & $1/2$     & $1/2$\\
    $B(t,h_1,h_2)$    & $1$       & $1/2$     & $1/2$     & $1/2$     & $1/2$\\
    $DY(t,h_2,h_2^B)$ & $1$       & $1/2$     & $1/2$     & $1/2$     & $1/2$\\
    \bottomrule
  \end{tabular}
\end{center}

The example is summarized in Fig.~\ref{f.1} where each arrow corresponds to one of the decision variables.
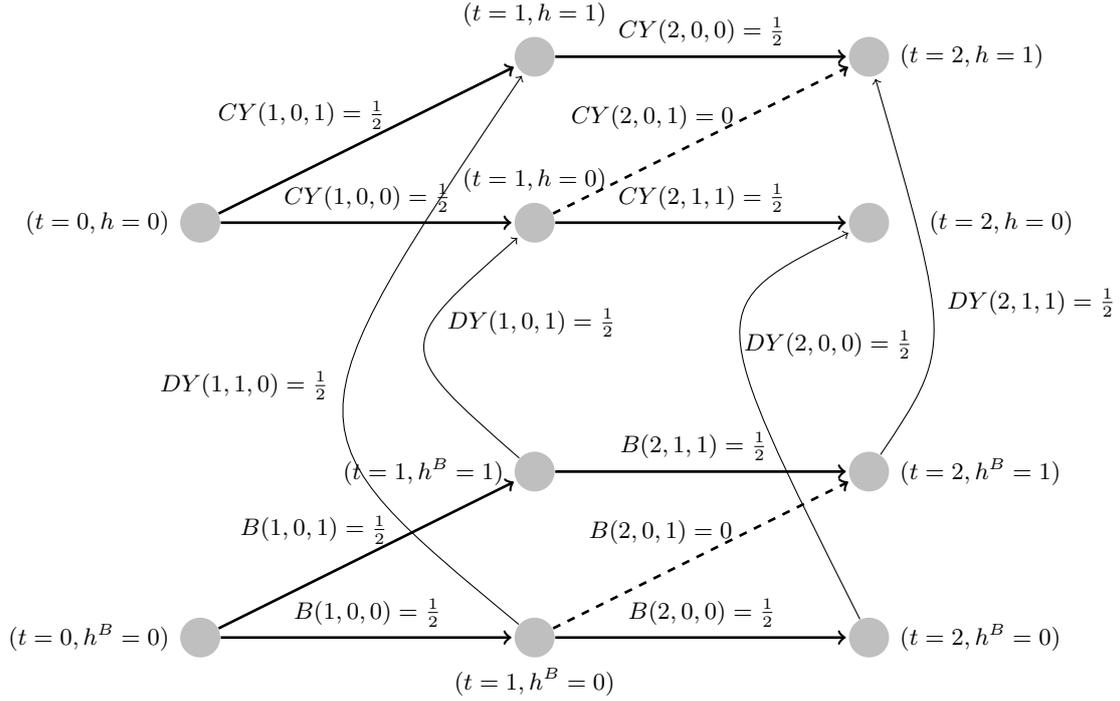
\begin{figure}
\label{f.1}
  \begin{center}
  \begin{tikzpicture}[scale=1.1, shorten >=1pt,->]
    \tikzstyle{vertex}=[circle,fill=black!25,minimum size=15pt,inner sep=0pt]

    \node[vertex,label=left:{\footnotesize{$(t=0,h=0)$}}] (c0) at (0,5) {};
    \node[vertex,label={\footnotesize{$(t=1,h=0)$}}] (c1) at (4,5) {};
    \node[vertex,label={\footnotesize{$(t=1,h=1)$}}] (c2) at (4,7) {};
    \node[vertex,label=right:{~~~\footnotesize{$(t=2,h=0)$}}] (c3) at (8,5) {};
    \node[vertex,label=right:{\footnotesize{$(t=2,h=1)$}}] (c4) at (8,7) {};

    \draw[line width= 1pt] (c0) -- (c1) node[pos=0.5,above]{\footnotesize{$CY(1,0,0)=\frac{1}{2}$}};
    \draw[line width= 1pt] (c0) -- (c2) node[pos=0.5,above]{\footnotesize{$CY(1,0,1)=\frac{1}{2}$~~~~~~~~~~~~~~~~~}};
    \draw[line width= 1pt] (c1) -- (c3) node[pos=0.5,above]{\footnotesize{$CY(2,1,1)=\frac{1}{2}$}};
    \draw[line width= 1pt] (c2) -- (c4) node[pos=0.5,above]{\footnotesize{$CY(2,0,0)=\frac{1}{2}$}};

    \draw[line width= 1pt, dashed] (c1) -- (c4) node[pos=0.5,above]{\footnotesize{$CY(2,0,1)=0$}~~~~~~~~~~~~~};

    \node[vertex,label=left:{\footnotesize{$(t=0,h^{B}=0)$}}] (b0) at (0,0) {};
    \node[vertex,label=below:{\footnotesize{$(t=1,h^{B}=0)$}}] (b1) at (4,0) {};
    \node[vertex,label=left:{\footnotesize{$(t=1,h^{B}=1)$}}] (b2) at (4,2) {};
    \node[vertex,label=right:{\footnotesize{$(t=2,h^{B}=0)$}}] (b3) at (8,0) {};
    \node[vertex,label=right:{\footnotesize{$(t=2,h^{B}=1)$}}] (b4) at (8,2) {};

    \draw[line width= 1pt] (b0) -- (b1) node[pos=0.5,above]{\footnotesize{$B(1,0,0)=\frac{1}{2}$}};
    \draw[line width= 1pt] (b0) -- (b2) node[pos=0.5,above]{\footnotesize{$B(1,0,1)=\frac{1}{2}$}~~~~~~~~~~~~~~};
    \draw[line width= 1pt] (b1) -- (b3) node[pos=0.5,above]{\footnotesize{$B(2,0,0)=\frac{1}{2}$}};
    \draw[line width= 1pt] (b2) -- (b4) node[pos=0.5,above]{\footnotesize{$B(2,1,1)=\frac{1}{2}$}~~~};
    \draw[line width= 1pt, dashed] (b1) -- (b4) node[pos=0.5,above]{\footnotesize{$B(2,0,1)=0$}~~~~~~~~~~~};


   \draw (b1) .. controls (1,2.5) .. (c2) node[pos=0.5,above]{\footnotesize{$DY(1,1,0)=\frac{1}{2}$}~~~~~~~~~~~~~~~~~~~~~~~~~};
   \draw (b2) .. controls (2.3,3.5) .. (c1) node[pos=0.5,above]{~~~~~~~~~~~~~~~~~~~~~~\footnotesize{$DY(1,0,1)=\frac{1}{2}$}};
   \draw (b3) .. controls (6,4) .. (c3) node[pos=0.4,above]{~~~~~~~~~~~~~~~~~\footnotesize{$DY(2,0,0)=\frac{1}{2}$}};
   \draw (b4) .. controls (9,3.5) .. (c4) node[pos=0.5,above]{~~~~~~~~~~~~~~~~~~~~\footnotesize{$DY(2,1,1)=\frac{1}{2}$}};
    \end{tikzpicture}
  \end{center}
  \caption{Example of non-integer point.}
\end{figure}

One can check that the example is a feasible solution (a point in the polytope).
Indeed, the flow conditions are satisfied,
as well as the equations linking the dummy variables $DY$ and the $CY$'s and $B$'s (Equations~\eqref{eq.3} and~\eqref{eq.5a}).

To argue that the point $P$ is indeed a vertex of the polytope, we show that for every line with non-zero direction vector $v=(x_0,\ldots,x_{14})$ and for every $\varepsilon>0$, either $P+\varepsilon v$ or $P-\varepsilon v$ is outside the polytope.
Every coordinate $x_i$ of $v$ corresponds, uniquely, to a variable $B(\cdot)$, $CY(\cdot)$, or $DY(\cdot)$.

First observe that if $x_i$ is the coordinate related to a variable that is either $0$ or $1$ in $P$, then $x_i=0$, as otherwise, for any $\varepsilon > 0$, either $P+\varepsilon v$ or $P-\varepsilon v$ would be outside of the polytope.
Hence, the only $x_i$ that may be non-zero, are those for which the coordinate $i$ in $P$ is in the open interval $(0,1)$.

In our example, every equation involves at most $2$ variables on each side of the equality, one of them being either $0$ or $1$. Hence the implications written below are forced by the previous observation.
Assume, for instance, that the coefficient $x_i$ corresponding to $B(2,1,1)$ in $v$ is negative.
\begin{itemize}
  \item Then, by the flow constraints (Equation~\eqref{eq.4}), the coefficient of $B(1,0,1)$ is negative.
  \item Then, by the flow constraints, the coefficient of $B(1,0,0)$ is positive.
  \item Then, by the flow constraints, the coefficient of $B(2,0,0)$ is positive.
\end{itemize}

Now, using the equations that link the variables $B$ and $DY$, we obtain that the the coefficient of $DY(2,1,1)$ is positive, which implies that
\begin{itemize}
  \item the coefficient of $CY(2,1,1)$ in $v$ is positive;
  \item then, by the flow constraints, the coefficient of $CY(1,0,1)$ is positive;
  \item then, by the flow constraints, the coefficient of $CY(1,0,0)$ is negative;
  \item then, by the flow constraints, the coefficient of $CY(2,0,0)$ is negative.
\end{itemize}
Observe now that this implies that the coefficient of $DY(2,0,0)$ has to be negative. However, let us now look at the coefficients of $DY(1,0,1)$ and the one corresponding to $DY(1,1,0)$.

If we use the links between the variables $DY$ and $B$, the coefficients corresponding to the variables $DY(1,0,1)$ and $DY(1,1,0)$ in $v$ have to be negative and positive respectively.
However, if we look at the equations linking the variables $DY$ and $CY$, the coefficients should have the opposite sign.
Thus, these coefficients should be zero, implying that all the other coefficients have to be $0$, which shows that no non-zero vector $v$ exists.

The first coefficient involved in the argument was the one involving the variable $B(2,1,1)$.
Since the implications described here involve all the non-zero variables of the point, and the implications are reversible, the result now follows.

\subsection{Sufficient conditions for integrality}
\label{subsection:avoidingnonintegralpoints}
In light of the above result, we now present some sufficient conditions on the objective function~\eqref{eq.of_1}--\eqref{eq.of_3}, that guarantee that either the linear relaxation of the integer program finds an integral point as a solution, or that there is an integral point in the optimal face and a procedure to find it.

\begin{proposition}
\label{p.1}
  Consider the IP model from Sect.~\ref{sec:modeldesc}.
  Assume that,
  \begin{enumerate}[(i)]
    \item \label{cond.1} for every $t,d,h_2,h_2',h_2^{B},h_2'^{B}$ such that $h_2\leq h_2'$ and $h_2^{B}\leq h_2'^{B}$, 		
	  \begin{equation}
      \label{e.1}
	    D_{\textnormal{expdam}}(t,d,h_{2}',h_{2}^{B}) + D_{\textnormal{expdam}}(t,d,h_{2},h_{2}'^{B})
	    \geq
	    D_{\textnormal{expdam}}(t,d,h_{2},h_{2}^{B}) + D_{\textnormal{expdam}}(t,d,h_{2}',h_{2}'^{B})
%
%
      \end{equation}
    \item \label{cond.2} for every $t,h_1^B,h_1'^B,h_2^B,h_2'^B$ such that $h_1^B\leq h_1'^B$ and $h_2^B\leq h_2'^B$,
	  \begin{equation}
      \label{eq.5}
	    B_{\textnormal{cost}}(t,h_1^B,h_2'^B)+B_{\textnormal{cost}}(t,h_1'^B,h_2^B)\geq B_{\textnormal{cost}}(t,h_1^B,h_2^B)+B_{\textnormal{cost}}(t,h_1'^B,h_2'^B)
	  \end{equation}
	\item \label{cond.3} for every $t,d,h_1,h_1',h_2,h_2'$ such that $h_1\leq h_1'$ and $h_2\leq h_2'$,
	  \begin{equation}\label{eq.6}
	    D_{\textnormal{cost}}(t,d,h_1,h_2')+D_{\textnormal{cost}}(t,d,h_1',h_2)\geq D_{\textnormal{cost}}(t,d,h_1,h_2)+D_{\textnormal{cost}}(t,d,h_1',h_2')\;.
	  \end{equation}
  \end{enumerate}
  Then, there is an optimal solution of the linear relaxation of the IP model in Sect.~\ref{sec:modeldesc} with integer coordinates.
\end{proposition}

Note that the term $\left(B_{\text{cost}}(t,h_1^B,h_2^B)+B_{\text{expdam}}(t,h_2^B)\right)$ from Equation~\eqref{eq.of_3} does not appear in condition~\eqref{cond.2} as it appears in both sides of the inequality.
\begin{proof}[Proof of Proposition~\ref{p.1}]
  The problem from Sect.~\ref{sec:modeldesc} can be thought of as several intertwined min-cost flow problems (see Sect.~\ref{sec:abstraction}), one for each dike, and one for the barrier.
  We say that a path in a graph $v_1 e_1 v_2 e_2 \ldots v_n$ with vertices $v_1 v_2 \ldots v_n$ and edges $e_1 e_2 \ldots e_{n-1}$, is a \emph{flow path} when the flow through each edge is the same.
  In our case, the vertices of the graph represent heights.
	
  Let $x_0$ be a solution point given by the linear relaxation, and assume it is non-integral.
  Using the monotone relations \eqref{eq.5} and \eqref{eq.6}, the paths of the non-zero flows that $x_0$ defines for each of the dikes and the barrier can be assumed to be completely ordered (as otherwise, the flow values on the edges might be modified while maintaining the value of the in-flow and out-flow at each vertex while not increasing the objective function).
  That is to say, we obtain a layered flow: a flow path $v_1 e_1 v_2 e_2 \ldots v_n$ with height profile $v_1 v_2 \ldots v_n$ is above a flow path $w_1 e'_1 w_2 e'_2 \ldots w_n$ with height profile $w_1 w_2 \ldots w_n$ when $v_i\geq w_i$ for all $i$ (i.e., no two flow-paths strictly cross between two layers of vertices corresponding to two different consecutive times).
  In particular, for each of the dikes $d$, we can talk about a top path $U_d$ (the height profile being always larger or equal than all the other height profiles), and a bottom path $L_d$, whose heights are smaller or equal than all the other height profiles.
  There is also a top $U_B$ and bottom $L_B$ paths for the flow of the barrier.
	%
	
  Observe that, as $x_0$ is non-integral, at least one of the variables $DY$ is non-integral (either not equal to zero or not equal to one).


%
%
%
%
%


	
  Let $DY_{\text{min}}$ be the minimal distance of the non-integral variables to either $0$ or $1$.
  Using~\eqref{e.1} as a guideline repeatedly, we modify the variables $DY$ from $x_0$ to create a new feasible solution~$x_1$ in which the variables $DY(t,i,h_2,h_2^B)$ are ``untangled''.
  That is: given $t,i,h_2,h_2',h_2^B,h_2'^B$ such that
  \[
    \begin{cases}
      h_2\leq h_2', h_2^B\leq h_2'^B \\
      1-DY_{\min}\geq DY(t,i,h_2,h_2^B)\geq DY_{\min}\\
      1-DY_{\min}\geq DY(t,i,h_2',h_2^B)\geq DY_{\min}\\
      1-DY_{\min}\geq DY(t,i,h_2,h_2'^B)\geq DY_{\min}\\
      1-DY_{\min}\geq DY(t,i,h_2',h_2'^B)\geq DY_{\min}
   \end{cases}
  \]
  then, by modifying
  \begin{equation}
  \label{eq:reassigningflow}
    \begin{cases}
      DY(t,i,h_2',h_2^B)\to DY(t,i,h_2',h_2^B)-DY_{\min} \\
      DY(t,i,h_2,h_2'^B)\to DY(t,i,h_2,h_2'^B)-DY_{\min}\\ 
      DY(t,i,h_2,h_2^B)\to DY(t,i,h_2,h_2^B)+DY_{\min}\\ 
      DY(t,i,h_2',h_2'^B)\to DY(t,i,h_2',h_2'^B)+DY_{\min} 
    \end{cases}
  \end{equation}
  and keeping the other values of solution $x_0$, we obtain a new feasible solution $x_1$ as good as $x_0$.
  In particular, by repeated application of the argument leading to \eqref{eq:reassigningflow}, we can assume that
  \begin{equation*}
    DY_{x_1}(t,i,h_2(U_i),h_2^{B}(U_B)) = \min\left\{\sum_{h_2} DY_{x_0}(t,i,h_2,h_2^{B}(U_{B})),\sum_{h_2^B} DY_{x_0}(t,i,h_2(U_{i}),h_2^{B})\right\}
  \end{equation*}
  and that
  \begin{equation*}
    DY_{x_1}(t,i,h_2(L_i),h_2^{B}(L_{B})) = \min\left\{\sum_{h_2} DY_{x_0}(t,i,h_2,h_2^{B}(L_{B})),\sum_{h_2^B} DY_{x_0}(t,i,h_2(L_i),h_2^{B})\right\},
  \end{equation*}
  while the remaining variables of $x_0$ are kept equal in $x_1$.
  As the reassignment preserves the flow constraints, $x_1$ remains feasible. By~\eqref{e.1}, $x_1$ has the same objective value as $x_0$, since $x_0$ is optimal.

  Let $F_{\min}$ be the minimal difference to $0$ or $1$ of the flow through each $L_{d}, U_{d}$ for every dike~$d$ and $L_{B}$ or $U_{B}$, which can be assumed to be the minimal value of
  \[
    \min_{t,i}\left\{DY_{x_1}(t,i,h_2(U_{i}),h_2^{B}(U_B)), DY_{x_1}(t,i,h_2(L_{i}),h_2^{B}(L_{B}))\right\} \enspace.
  \]
  As we shall see, $F_{\min}$ is the minimal amount of flow which is reassigned between the upper and lower paths.
%
%
%
%

%
%
	
  We note that $x_1$ is not a vertex of the polytope.
  Indeed, for any dike $d$, we can pair up $L_d\leftrightarrow L_B$ and $U_d\leftrightarrow U_{B}$.
  Using \eqref{eq.5} and \eqref{eq.6}, this pairing is well defined and consistent.
  In particular, we can redirect an $\varepsilon$ amount of flow---where $0<\varepsilon \leq F_{\min}$---from each of the $L_d$ to~$U_d$ and from $L_B$ to $U_B$, or vice versa (the redirection of the flow should be done on each of the paths simultaneously, either from upper to lower paths, or from lower to upper ones).
  Since there exists a $d$ (or $B$) for which the paths $L_d$ and $U_d$ differ, this flow-redirection by $\varepsilon$ gives a different point on the polytope of feasible points and shows that $x_1$ is not a vertex of the polytope.

  Furthermore, for every $\varepsilon>0$, the mentioned flow redirection should give the same value of the objective function (since otherwise $x_0$ would not have been an optimal solution).
  Hence we can choose to redirect the flow at our convenience; we redirect it so that the edge whose flow value is $F_{\min}$ becomes either $0$ or $1$ (depending on whether its value is closer to $0$ or to~$1$, if $F_{\min}=1/2$, we arbitrarily redirect the flow either way).
  In particular, we obtain a new solution~$x_2$ where the number of edges with non-integral flow has been reduced by at least one.
  This procedure can be iterated until no non-integral flows are found.
  Therefore, an integral vertex of the polytope in the optimal face of the linear relaxation of the integer program is found.
%
%
\end{proof}

\begin{corollary}
\label{coro1}
  The conclusion of Proposition~\ref{p.1} also holds if we assume conditions~\eqref{cond.2} and~\eqref{cond.3}, and condition~\eqref{cond.1} on the objective function is replaced by
  \begin{enumerate}[(i')]
    \item \label{cond.1p} For each dike $d \in D$, either
	  \begin{equation}
        \label{eq.2}
	    D_{\textnormal{expdam}}(t,d,h_{2}',h_{2}^{B}) + D_{\textnormal{expdam}}(t,d,h_{2},h_{2}'^{B})
	    \leq 	
	    D_{\textnormal{expdam}}(t,d,h_{2},h_{2}^{B}) + D_{\textnormal{expdam}}(t,d,h_{2}',h_{2}'^{B})
%
%
	  \end{equation}
     for every $t,h_2,h_2',h_2^{B},h_2'^{B}$ such that $h_2\leq h_2'$ and $h_2^{B}\leq h_2'^{B}$, or
     \begin{equation}
     \label{e.11}
       D_{\textnormal{expdam}}(t,d,h_{2}',h_{2}^{B}) + D_{\textnormal{expdam}}(t,d,h_{2},h_{2}'^{B})
	   \geq 	
       D_{\textnormal{expdam}}(t,d,h_{2},h_{2}^{B}) + D_{\textnormal{expdam}}(t,d,h_{2}',h_{2}'^{B})
%
%
	\end{equation}
    for every $t,h_2,h_2',h_2^{B},h_2'^{B}$ such that $h_2\leq h_2'$ and $h_2^{B}\leq h_2'^{B}$.
  \end{enumerate}
\end{corollary}
\begin{proof}
  The argument of the proof of Proposition~\ref{p.1} should be modified as follows.
  Observe that the layering of the flow-paths can be maintained due to conditions~\eqref{cond.2} and~\eqref{cond.3}.
  Then, the flow path pairing that allows for the flow reassignment of the second part of the proof can be modified as follows.
  The dike $d$ uses the pairing
  \[
    L_d\leftrightarrow L_B \text{ and } U_d\leftrightarrow U_B
  \]
  if part~\eqref{e.11} of condition~\eqref{cond.1p} is satisfied, and it is exchanged by the new pairing
  \[
    L_d\leftrightarrow U_B \text{ and } U_d\leftrightarrow L_B \; .
  \]
  if part~\eqref{eq.2} is satisfied.
  In each of the cases we use either the modification on the $DY$ due to~\eqref{eq:reassigningflow} or due to~\eqref{eq:reassigningflow} where the $-DY_{min}$ and $+DY_{min}$ are exchanged.

  If the solution is non-integral, there exists a flow reassignment  between the paths paired with~$U_B$, and the paths paired with $L_B$.
  Mutatis mutandis, the reassigning flow argument carries over to this new case.
\end{proof}

To provide some intuition, observe that all inequalities appearing in conditions~\eqref{cond.1}-~\eqref{cond.3} are of a similar form: for $a \leq a'$ and $b \leq b'$ we have that some function $c(\cdot,\cdot)$ satisfies $c(a,b) + c(a',b') \leq c(a,b') + c(a',b)$.
Such an inequality for $c$ is naturally satisfied (in fact with equality) if $c$ is of the form $c(x,y) = f(y) - f(x) + c_0$ for some function $f$ and constant $c_0$.
In the context of conditions~\eqref{cond.2} and~\eqref{cond.3} from Proposition~\ref{p.1}, such a form is somewhat reasonable to expect: the cost of rising a dike from level $x$ to $y$ compares to the cost of rising the dike from level $0$ to $y$, minus the effort already made to rise it from $0$ to $x$, plus perhaps some inefficiency overhead $c_0$.

\subsection{Computational results}
Conditions~(i'),~\eqref{cond.2} and~\eqref{cond.3} from Corollary~\ref{coro1} and Proposition~\ref{p.1} have been implemented and
tested for the most recent data on flood probabilities, damage and investment costs, and the results confirm that they are often met.


In the first column of the following tables we specify the years that we used in our study: 5 year periods until 2100 and 10 year periods after 2100.
In the first row we specify the specific dike rings.
The description of the dike rings around Lake IJssel and the IJsseldelta is as follows (the numbers are also used in Fig.~\ref{fig:dikerings}):\\

\begin{tabular}{lll}
  zwf = Zuid-West Friesland = 6.4 &\qquad\qquad& nop =   Noord-Oost Polder = 7.1\\
  nfl = Noord-Oost Flevoland = 8.1 &\qquad\qquad& wfn =  West-Friesland Noord = 13.2\\
  wie = Wieringen = 12.1 &\qquad\qquad& ijd =  IJsseldelta = 11.1\\
  mas = Mastenbroek = 10.1 &\qquad\qquad& vol =  Vollenhove = 9.1\\
  sal = Salland = 53.1 &\qquad\qquad& ovl =  Oost-Veluwe = 52.1
\end{tabular}

\begin{figure}[ht!]
\centering
  \includegraphics[scale=0.5]{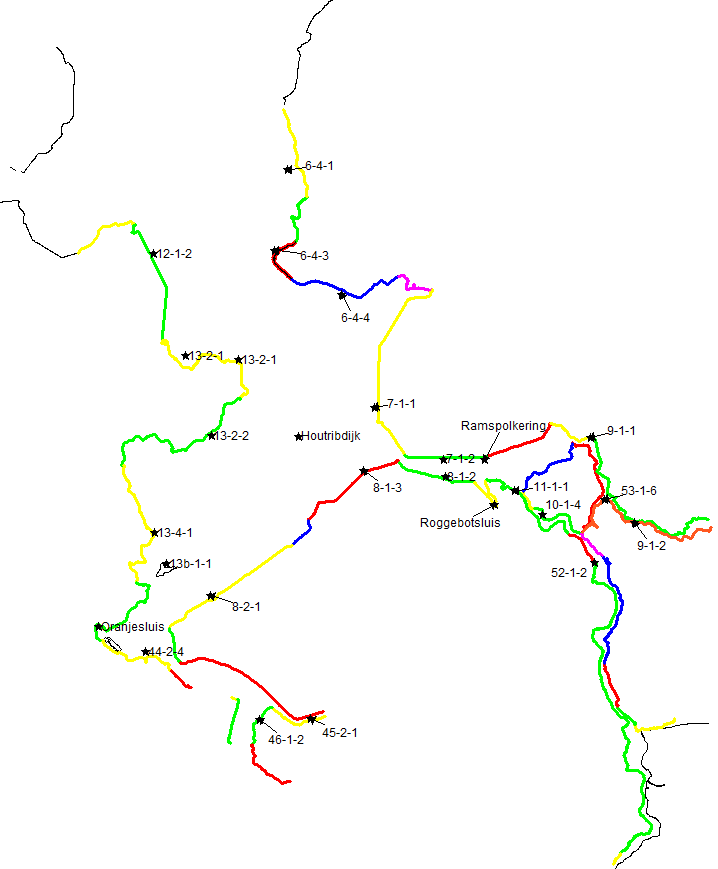}
    \caption{Dike rings around Lake IJssel and the IJsseldelta in The Netherlands.}
  \label{fig:dikerings}
\end{figure}

 As safety levels we included 14 levels for the dike rings and also 14 levels for the Afsluitdijk. As a result, for each dike ring  $\frac{1}{2}\cdot14\cdot15=105$ combinations of safety levels for both the dike rings and the Afsluitdijk could be evaluated. Hence, in total we tested $105\cdot 105=11025$ instances (numbers in the second column). The rest of numbers in the following tables correspond to the number of instances that fulfill the corresponding condition and are calculated for all dike rings (first numerical row) and year (first column) combinations. Note that the year 2020 is the initial year considered, thus the conditions of Corollary~\ref{coro1} and Proposition~\ref{p.1} are trivially satisfied for each dike ring (zwf, nop, nfl, wfn,  wie, ijd, mas, vol, sal, ovl) and barrier (Afsluitdijk); the ``2020'' row is hence added for comparison purposes.

We ran the simulation with pumps on the Afsluitdijk and without additional strengthening measures of the Afsluitdijk (STA-program).

\vspace{1cm}

\tiny{
\begin{center}
\begin{tabular}{ |c|c|c|c|c|c|c|c|c|c|c| }
 \hline
 \multicolumn{11}{|c|}{Condition~(i') from Corollary~\ref{coro1}} \\
 \hline
year & zwf & nop & nfl & wfn & wie & ijd & mas & vol & sal & ovl \\ \hline
2020 & 11025 & 11025 & 11025 & 11025 & 11025 & 11025 & 11025 & 11025 & 11025 & 11025\\
2021 & 10930 & 11025 & 9484 & 10935 & 11025 & 10604 & 10105 & 10484 & 10634 & 10869\\
2026 & 10815 & 11022 & 9494 & 10909 & 11025 & 10602 & 10117 & 10488 & 10639 & 10869\\
2031 & 10750 & 11023 & 9510 & 10874 & 11025 & 10600 & 10125 & 10491 & 10645 & 10856\\
2036 & 10936 & 11025 & 9530 & 10857 & 11025 & 10598 & 10130 & 10496 & 10656 & 11024\\
2041 & 10970 & 11025 & 9553 & 10832 & 11025 & 10596 & 10138 & 10499 & 10670 & 10996\\
2046 & 10992 & 11025 & 9588 & 10821 & 11025 & 10594 & 10148 & 10503 & 10683 & 10996\\
2051 & 11000 & 11025 & 9618 & 10811 & 11025 & 10594 & 10148 & 10506 & 10683 & 10963\\
2056 & 11010 & 11025 & 9666 & 10807 & 11025 & 10594 & 10148 & 10511 & 10683 & 10875\\
2061 & 11017 & 11025 & 9674 & 10807 & 11025 & 10594 & 10156 & 10515 & 10683 & 10809\\
2066 & 11021 & 11025 & 9677 & 10807 & 11025 & 10594 & 10170 & 10517 & 10683 & 10736\\
2071 & 11024 & 11025 & 9698 & 10803 & 11025 & 10594 & 10184 & 10520 & 10694 & 10676\\
2076 & 11025 & 11025 & 9715 & 10799 & 11025 & 10594 & 10194 & 10526 & 10708 & 10619\\
2081 & 11025 & 11025 & 9762 & 10803 & 11025 & 10594 & 10198 & 10530 & 10709 & 10566\\
2086 & 11025 & 11025 & 9777 & 10803 & 11025 & 10600 & 10205 & 10533 & 10713 & 10510\\
2091 & 11025 & 11025 & 9812 & 10803 & 11025 & 10600 & 10222 & 10537 & 10721 & 10487\\
2096 & 11025 & 11025 & 9847 & 10807 & 11025 & 10600 & 10239 & 10542 & 10736 & 10383\\
2101 & 11025 & 11025 & 9895 & 10807 & 11025 & 10600 & 10266 & 10542 & 10751 & 8185\\
2111 & 11025 & 11025 & 9927 & 10809 & 11025 & 10600 & 10297 & 10531 & 10763 & 8051\\
2121 & 11025 & 11025 & 9993 & 10894 & 11025 & 10601 & 10297 & 10506 & 10771 & 7986\\
2131 & 11025 & 11025 & 10020 & 10954 & 10934 & 10606 & 10297 & 10498 & 10771 & 7937\\
2141 & 10975 & 11025 & 10098 & 10988 & 10752 & 10612 & 10304 & 10472 & 10776 & 7920\\
2151 & 10854 & 11025 & 10197 & 10987 & 10492 & 10633 & 10311 & 10470 & 10799 & 7876\\
2161 & 10636 & 11025 & 10236 & 10953 & 10206 & 10633 & 10316 & 10439 & 10825 & 7858\\
2171 & 10207 & 11025 & 10289 & 10912 & 9751 & 10654 & 10331 & 10388 & 10831 & 7842\\
2181 & 9608 & 11025 & 10365 & 10783 & 9205 & 10696 & 10369 & 10273 & 10847 & 7796\\
2191 & 9257 & 11025 & 10447 & 10436 & 9114 & 10747 & 10423 & 10120 & 10857 & 7755\\
2201 & 9153 & 11025 & 10456 & 9881 & 9114 & 10779 & 10492 & 9919 & 10874 & 7727\\
2211 & 8533 & 11025 & 10473 & 9299 & 8477 & 10815 & 10603 & 9745 & 10876 & 7607\\
2221 & 8359 & 11025 & 10520 & 9050 & 8477 & 10853 & 10681 & 9519 & 10874 & 7419\\
2231 & 7936 & 11025 & 10607 & 8802 & 8477 & 10868 & 10730 & 9261 & 10823 & 7384\\
2241 & 7517 & 11025 & 10640 & 8459 & 7840 & 10901 & 10773 & 7748 & 6573 & 7352\\
2251 & 7067 & 11025 & 10708 & 8251 & 7749 & 10921 & 10815 & 4884 & 5581 & 7339\\
2261 & 6579 & 11025 & 10745 & 7849 & 7293 & 10933 & 10846 & 3719 & 5222 & 7270\\
2271 & 5994 & 11025 & 10864 & 7501 & 6930 & 10954 & 10881 & 3403 & 5052 & 7276\\
2281 & 5475 & 11025 & 10935 & 7096 & 6930 & 10954 & 10923 & 3245 & 5021 & 7276\\
2291 & 4786 & 11025 & 10975 & 6653 & 6111 & 10973 & 10925 & 3212 & 5026 & 7253\\
\hline
\end{tabular}
\end{center}

\begin{center}
\begin{tabular}{ |c|c|}
 \hline
 \multicolumn{2}{|c|}{Condition~\eqref{cond.2} from Proposition~\ref{p.1}} \\
 \hline
year & afsluitdijk \\  \hline
2020 & 2380\\
2021 & 2016\\
2026 & 2016\\
2031 & 2016\\
2036 & 2016\\
2041 & 2016\\
2046 & 2016\\
2051 & 2016\\
2056 & 2016\\
2061 & 2016\\
2066 & 2016\\
2071 & 2016\\
2076 & 2016\\
2081 & 2016\\
2086 & 2016\\
2091 & 2016\\
2096 & 2016\\
2101 & 2016\\
2111 & 2016\\
2121 & 2016\\
2131 & 2016\\
2141 & 2016\\
2151 & 2016\\
2161 & 2016\\
2171 & 2016\\
2181 & 2016\\
2191 & 2016\\
2201 & 2016\\
2211 & 2016\\
2221 & 2016\\
2231 & 2016\\
2241 & 2016\\
2251 & 2016\\
2261 & 2016\\
2271 & 2016\\
2281 & 2016\\
2291 & 2016\\
\hline
\end{tabular}
\end{center}

\begin{center}
\begin{tabular}{ |c|c|c|c|c|c|c|c|c|c|c| }
 \hline
 \multicolumn{11}{|c|}{Condition~\eqref{cond.3} from Proposition~\ref{p.1}} \\
 \hline
year & zwf & nop & nfl & wfn & wie & ijd & mas & vol & sal & ovl \\ \hline
2020 & 2380 & 2380 & 2380 & 2380 & 2380 & 2380 & 2380 & 2380 & 2380 & 2380\\
2021 & 2162 & 2188 & 2016 & 2089 & 2104 & 2018 & 2217 & 2180 & 2210 & 2085\\
2026 & 2162 & 2188 & 2016 & 2089 & 2104 & 2018 & 2217 & 2180 & 2210 & 2085\\
2031 & 2162 & 2188 & 2016 & 2089 & 2104 & 2018 & 2217 & 2180 & 2210 & 2085\\
2036 & 2162 & 2188 & 2016 & 2089 & 2104 & 2018 & 2217 & 2180 & 2210 & 2085\\
2041 & 2162 & 2188 & 2016 & 2089 & 2104 & 2018 & 2217 & 2180 & 2210 & 2085\\
2046 & 2162 & 2188 & 2016 & 2089 & 2104 & 2018 & 2217 & 2180 & 2210 & 2085\\
2051 & 2162 & 2188 & 2016 & 2089 & 2104 & 2018 & 2217 & 2180 & 2210 & 2085\\
2056 & 2162 & 2188 & 2016 & 2089 & 2104 & 2018 & 2217 & 2180 & 2210 & 2085\\
2061 & 2162 & 2188 & 2016 & 2089 & 2104 & 2018 & 2217 & 2180 & 2210 & 2085\\
2066 & 2162 & 2188 & 2016 & 2089 & 2104 & 2018 & 2217 & 2180 & 2210 & 2085\\
2071 & 2162 & 2188 & 2016 & 2089 & 2104 & 2018 & 2217 & 2180 & 2210 & 2085\\
2076 & 2162 & 2188 & 2016 & 2089 & 2104 & 2018 & 2217 & 2180 & 2210 & 2085\\
2081 & 2162 & 2188 & 2016 & 2089 & 2104 & 2018 & 2217 & 2180 & 2210 & 2085\\
2086 & 2162 & 2188 & 2016 & 2089 & 2104 & 2018 & 2217 & 2180 & 2210 & 2085\\
2091 & 2162 & 2188 & 2016 & 2089 & 2104 & 2018 & 2217 & 2180 & 2210 & 2085\\
2096 & 2162 & 2188 & 2016 & 2089 & 2104 & 2018 & 2217 & 2180 & 2210 & 2085\\
2101 & 2162 & 2188 & 2016 & 2089 & 2104 & 2018 & 2217 & 2180 & 2210 & 2085\\
2111 & 2162 & 2188 & 2016 & 2089 & 2104 & 2018 & 2217 & 2180 & 2210 & 2085\\
2121 & 2162 & 2188 & 2016 & 2089 & 2104 & 2018 & 2217 & 2180 & 2210 & 2085\\
2131 & 2162 & 2188 & 2016 & 2089 & 2104 & 2018 & 2217 & 2180 & 2210 & 2085\\
2141 & 2162 & 2188 & 2016 & 2089 & 2104 & 2018 & 2217 & 2180 & 2210 & 2085\\
2151 & 2162 & 2188 & 2016 & 2089 & 2104 & 2018 & 2217 & 2180 & 2210 & 2085\\
2161 & 2162 & 2188 & 2016 & 2089 & 2104 & 2018 & 2217 & 2180 & 2210 & 2085\\
2171 & 2162 & 2188 & 2016 & 2089 & 2104 & 2018 & 2217 & 2180 & 2210 & 2085\\
2181 & 2162 & 2188 & 2016 & 2089 & 2104 & 2018 & 2217 & 2180 & 2210 & 2085\\
2191 & 2162 & 2188 & 2016 & 2089 & 2104 & 2018 & 2217 & 2180 & 2210 & 2085\\
2201 & 2162 & 2188 & 2016 & 2089 & 2104 & 2018 & 2217 & 2180 & 2210 & 2085\\
2211 & 2162 & 2188 & 2016 & 2089 & 2104 & 2018 & 2217 & 2180 & 2210 & 2085\\
2221 & 2162 & 2188 & 2016 & 2089 & 2104 & 2018 & 2217 & 2180 & 2210 & 2085\\
2231 & 2162 & 2188 & 2016 & 2089 & 2104 & 2018 & 2217 & 2180 & 2210 & 2085\\
2241 & 2162 & 2188 & 2016 & 2089 & 2104 & 2018 & 2217 & 2180 & 2210 & 2085\\
2251 & 2162 & 2188 & 2016 & 2089 & 2104 & 2018 & 2217 & 2180 & 2210 & 2085\\
2261 & 2162 & 2188 & 2016 & 2089 & 2104 & 2018 & 2217 & 2180 & 2210 & 2085\\
2271 & 2162 & 2188 & 2016 & 2089 & 2104 & 2018 & 2217 & 2180 & 2210 & 2085\\
2281 & 2162 & 2188 & 2016 & 2089 & 2104 & 2018 & 2217 & 2180 & 2210 & 2085\\
2291 & 2162 & 2188 & 2016 & 2089 & 2104 & 2018 & 2217 & 2180 & 2210 & 2085\\
\hline
\end{tabular}
\end{center}
}

\newpage

\normalsize
\section{Alternative approaches}
\label{sec:altapproach}
A feasible solution to the IP presented in Sect.~\ref{sec:modeldesc} can be interpreted as a choice of height $h^d(t)$ for each dike segment at each time period $t$, and a height $h^b(t)$ of the barrier dam.
Abstractly, the cost of these height series can be written as a sum of cost terms which depend only on the `upgrade' done in period $t$ to segment $d$ (i.e., a heightening of the dike, or merely the maintenance cost); we denote this by $\mathsf{cost}^d(h^d(t-1),h^d(t),t)$ for segment $d$, and by $\mathsf{cost}^b(h^b(t-1), h^b(t),t)$ for the barrier.
Finally, there is also an expected damage cost for upgrading the dike and barrier to heights $h^d(t)$ and $h^b(t)$ in period $t$, denoted by $\mathsf{dam}^{d,b}(h^b(t),h^d(t),t)$.
The problem modeled in Sect.~\ref{sec:modeldesc} can thus be written in the following way:
\begin{align}
\label{eqn:ipremodel}
  \min \Big\{\sum_{t \in [T]} \mathsf{cost}^b&(h^b(t-1), h^b(t),t) + \sum_{d \in D}  \mathsf{cost}^d(h^d(t-1),h^d(t), t)  + \mathsf{dam}^{d,b}(h^b(t),h^d(t),t)  \\
\textnormal{s.t.} \quad & h^d(t) \in H_D, h^b(t) \in H_B  \text{ for } d \in D, t \in T \\
& h^d(t) \geq h^d(t-1) \text{ for } d \in D, t \in T \\
& h^b(t) \geq h^b(t-1) \text{ for } t \in T\Big\}
\end{align}
The linear relaxation of the IP model presented in Sect.~\ref{sec:modeldesc} can be solved in time polynomial in $|D|,|T|,|H_D|$, and $|H_B|$.
However, in general there is no guarantee that the returned solution is integral, see Sect.~\ref{sec:integrability}.
In the next two sections we describe two different approaches to solving this problem.
Both approaches have the benefit of solving the integer problem exactly.
However, this comes at a cost: both approaches give a polynomial time algorithm only if one of the parameters is regarded as a constant.
The first approach is to solve the integer program by ways of a dynamic program.
The second approach comes down to enumerating all possible height profiles of the barrier dam, and for each profile solving shortest path problems on small graphs.

\subsection{Dynamic programming}\label{sec:DP}
There are two key observations to be made.
First, the second part of the objective function decomposes naturally into a sum of $|D|$ terms, each of which depends only on the barrier height and one segment. Secondly, for each time period the cost only depends on the dike/barrier heights at times $t-1$ and~$t$.
Together this allows us to solve the problem using a dynamic program.
The recursion will be on the time period.
We maintain a table which stores values $\mathsf{opt} (h^b, \mathbf{h^s}, t)$ for all $t \in T, h^b \in H_B, \mathbf{h^d} \in (H_D)^D$.
Their interpretation is as that $ \mathsf{opt} (h^b, \mathbf{h^d}, t)$ is equal to the minimum cost made, up to time~$t$, when the barrier and segments are of height $h^b$ and $\mathbf{h^d}$ at time period $t$ respectively.
We can compute the entries of this table by means of the following recursion:
\begin{align*}
  \mathsf{opt} (h^b,\mathbf{h^d}, t) = \min\Big\{ &\mathsf{opt}(h^b-i^b, \mathbf{h^d}-\mathbf{i^d},t-1) +\mathsf{cost}^b(h^b-i^b, h^b, t) +\\
                                                  &\mathsf{cost}(\mathbf{h^d}-\mathbf{i^d},\mathbf{h^d}, t) + \mathsf{dam}(h^b,\mathbf{h^d},t): \\
                                                  &\qquad h^b-i^b \in H_B, \mathbf{h^d} - \mathbf{i^d} \in (H_D)^{|D|} \Big\}
\end{align*}
It follows that each entry of the table can be computed in time $\mathcal O (|H_B| |H_D|^{|D|})$. Hence, all entries of the table can be filled in time $\mathcal{O}\big((|H_B| |H_D|^{|D|})^2 \cdot |T|\big)$.
Using the interpretation of $\mathsf{opt}(h^b, \mathbf{h^d},t)$, it follows that the optimum of \eqref{eqn:ipremodel} is equal to
\[
  \min_{h^b \in H_B, \mathbf{h^d} \in (H_D)^{|D|}}  \mathsf{opt}(h^b, \mathbf{h^d},T) \enspace .
\]

This shows the following result:
\begin{theorem}
  One can determine the optimal value of \eqref{eqn:ipremodel} in time $\mathcal{O}\big((|H_B| |H_D|^{|D|})^2 \cdot |T|\big)$.
\end{theorem}


\subsection{Shortest paths}
In the previous section we have seen an algorithm for computing the optimal dike/barrier height profiles which has polynomial runtime for a fixed number of dike segments, in this section we present a different algorithm, based on shortest paths, that runs in polynomial time when the number of possible barrier heights is fixed.
We present an algorithm that computes the optimal value of \eqref{eqn:ipremodel} in time
\[
\mathcal{O}\left(\overbrace{|D|}^{\textsf{\# segments}} \cdot \underbrace{(T\cdot |H_D|)^2}_{\textsf{complexity shortest path}} \cdot \overbrace{T^{|H_B|}}^{\textsf{\# barrier height profiles}} \right).
\]
To illustrate the basic idea we first discuss the algorithm for the setting of one dike segment and no barrier, we then add a barrier dam and from that the generalization to multiple dike segments and barriers easily follows.

\subsubsection{One dike segment, no barrier}
First consider the situation with only one dike segment and no barrier.
In this case the problem of minimizing the cost at time period $T$ becomes equivalent to finding a shortest $p$-$q$ path in the following graph.
The source $p=(0,0)$ is the initial height of the dike at time $0$.
Then, for each time $t \in \{1,\hdots,T\}$ and each possible height of the dike $h$, we define a node $(t,h)$.
Finally we define a sink node $q$.
The edges are defined as follows.
We first add an edge between $(0,0)$ and $(1,h)$ for each $h \in H_D$, with weight $\mathsf{cost}(0,h,1)$, similarly for each $t \in \{1,\hdots,T\}$ and height pair $h_1 \leq h_2$ there is an edge from $(t-1,h_1)$ to $(t,h_2)$ with weight $\mathsf{cost}(h_1,h_2,t)$ equal to the financial cost associated to the decision of raising the dike segment from height $h_1$ to $h_2$ in time period $t$.
Notice that since there is no barrier, we can assume that the expected damage cost $\mathsf{dam}(t,h)$ are incorporated in $\mathsf{cost}(h_1,h_2,t)$.
Finally, the nodes $(T,h)$ are all connected to the sink $q$.
In the figure below the incoming and outgoing arcs of a node $(t,h_2)$ are sketched for some $0<t<T$ and $h_2 \in H_D$.
One observes that, indeed, the shortest $p$-$q$ path corresponds to the best strategy of heightening this dike segment.

Recall, the shortest $p$-$q$ path in a graph $G = (V,E)$ with non-negative edge weights can be found in time $\mathcal O (|V|^2)$ using Dijkstra's algorithm.

\begin{center}
  \begin{tikzpicture}[scale=0.85, shorten >=1pt,->]
    \tikzstyle{vertex}=[circle,fill=black!25,minimum size=15pt,inner sep=0pt]

    \node[vertex,label=left:{$(t-1,h_2)$}] (n0) at (0,0) {};
    \node[vertex,label=left:{$(t-1,h_2-1)$}] (n1) at (0,-1.5) {};
    \node[vertex,label=left:{$(t-1,h_2-2)$}] (n2) at (0,-3) {};

    \node[vertex,label={$(t,h_2)$}] (m0) at (5,0) {};

    \node[vertex,label=right:{$(t+1,h_2)$}] (r0) at (10,0) {};
    \node[vertex,label=right:{$(t+1,h_2+1)$}] (r1) at (10,1.5) {};
    \node[vertex,label=right:{$(t+1,h_2+2)$}] (r2) at (10,3) {};

    \draw[line width= 1pt] (n0) -- (m0);
    \draw[line width= 1pt] (n1) -- (m0);
    \draw[line width= 1pt] (n2) -- (m0);
    \draw[line width= 1pt] (m0) -- (r0);
    \draw[line width= 1pt] (m0) -- (r1);
    \draw[line width= 1pt] (m0) -- (r2);

    \node at (2.5,0.25) {$\mathsf{cost}(t,h_1,h_2)$};
  \end{tikzpicture}
\end{center}

\subsubsection{One dike segment, a barrier}
\label{sec:CPB one dike one barrier}
We now consider the case of a single dike segment and a barrier.
The observation we need to make is that the total financial cost incurred by upgrading the dike segment from height $h_1$ to height $h_2$ in time period $t$ no longer only depend on the dike segment, they also depend on the height of the barrier at time point $t$.
This means that we cannot solve a shortest path problem for the barrier and dike segment separately: the costs on the dike segment graph depend on the path chosen in the barrier graph.

The key idea is that if we fix the height of the barrier at each time $t$, then we reduce to the previous setting where all the costs are known.
Hence, the optimal value of \eqref{eqn:ipremodel} can be found by minimizing over the possible height profiles $h^b(t)$ of the barrier over time, the minimum cost of a $p$-$q$ path in the network defined in the previous section (using the costs associated to $h^b(t)$) plus the cost of implementing height profile $h^b(t)$.
The outer minimization over the possible height profiles $h^b(t)$ is performed by enumeration, which takes time $\mathcal O(T^{|H_B|})$.
This means that the optimal investment strategy for both the dike segment and barrier can be found in time
\[
  \mathcal{O}\left((T\cdot|H_D|)^2 \cdot {T \choose |H_B|}\right) = \mathcal{O}\left((T\cdot|H_D|)^2 \cdot T^{|H_B|}\right) \enspace.
\]

\subsubsection{Multiple dike segments and a barrier}
The approach of the previous section easily generalizes to the setting of multiple dike segments and a barrier.
Once a height profile $h^b(t)$ of the barrier dike is fixed, the optimal height profiles of each of the different dike segments can be computed independently.
Hence the problem of finding the optimal investment strategy for multiple dike segments and a barrier can be solved in time
\[
  \mathcal{O}\left(|D| \cdot (T \cdot |H_D|)^2 \cdot  T^{|H_B|} \right).
\]
This approach generalizes to the setting of multiple barriers and dike segments (where the costs of a dike segment at time $t$ may depend on the height of several barriers).
The complexity will be of the form
\[
  \mathcal{O}\left(|D| \cdot (T\cdot|H_D|)^2 \cdot  T^{|H_B||B|} \right),
\]
where $|B|$ is the number of barriers.
One should note that the above approach assumes the same discretization in time of the barrier and dike segments. It seems reasonable to assume a coarser discretization for the barrier of say $T_B$ steps, this would reduce the above-mentioned formula to
\[
  \mathcal{O}\left(|D| \cdot (T\cdot|H_D|)^2 \cdot  (T_B)^{|H_B||B|} \right) \enspace .
\]

\section{An abstraction of the dike heightening problem}
\label{sec:abstraction}
In this section we present a natural abstract version of the dike heightening problem, which allows for several variations and questions, which we believe have not been considered in the literature before.
We believe that studying these variations may shed more light on the complexity of the dike height problem.

In the dike height problem we essentially have two directed graphs where each path in one of the two graphs (the one modeling the height of the barrier dam) influences the cost of arcs in the other graph.
It is not difficult to show that if we were to allow any kind of influence of the path in the one graph on the cost of arcs in the other graph, the problem would become $\mathsf{NP}$-hard.
Indeed, one can easily show that in this case the problem contains the problem of finding two vertex disjoint paths in a directed graph, which is $\mathsf{NP}$-complete~\cite{cpb:fortune}.

For this reason, we consider the following restricted problem.
\begin{definition}
  For $k\in \mathbb{N}$, a $k$-\emph{layered} graph is a directed graph $D=(V,A)$ such that $V$ is partitioned into \emph{layers} $V=V_0\cup V_1\cup \ldots \cup V_k\cup V_{k+1}$ such that each $a\in A$ is from $V_i$ to~$V_{i+1}$ for some $i=0,\ldots,k$, where $V_0$ and $V_{k+1}$ both consist of a single vertex and where $|V_1|=|V_2|=\cdots=|V_k|$.
  We denote the arcs between $V_i$ and $V_{i+1}$ by $A[V_i,V_{i+1}]$ and we refer to $|V_1|$ as the \emph{partition size}.
\end{definition}

With this in mind, we define the {\sc Minimum Intertwined Cost Path} problem as follows.
The problem takes as input $d+1$ $k$-layered graphs $G^1=(V^1,A^1),G^2=(V^2,A^2)$,$\ldots$,$G^{d+1}=(V^{d+1},A^{d+1})$ with partitions $V^j=V^{(j)}_1\cup \ldots \cup V^{(j)}_{k+1}$ and cost functions $c^j:A^j\to \mathbb{R}_{\geq 0}$ for $j=1,\ldots,d+1$, and for each $i=1,\ldots,k$ and $t=2,\ldots,d+1$ maps $m_i^t:V^{(t)}_i\times A^1[V^1_{i-1},V^1_{i}]\to \mathbb{R}_{\geq 0}$.

Given $d+1$ paths $P^1,P^2,\ldots,P^{d+1}$ with $P^j = (a_1^j,v_1^j,a_2^j,v_2^j,\hdots,a_k^j,v_k^j,a_{k+1}^j)$ from $V_0^{(j)}$ to $V^{(j)}_{k+1}$ with $a^j_i = (v_i^j,v_{i+1}^j)$ for $j = 1,\ldots,d+1$, we define the \emph{cost} of the $(d+1)$-tuple $(P^1;P^2,\ldots,P^{d+1})$ as
\[
  \text{cost}(P^1;P^2,\ldots,P^{d+1})=\sum_{i=1}^{k+1}\sum_{t=1}^{d+1} c^t(a^t_i)+\sum_{i=1}^{k+1}\sum_{t=2}^{d+1} m_i^t(v^t_i,a_i^1) \enspace .
\]
The objective is to compute the $(d+1)$-tuple of paths $(P^{1*};P^{2*},\ldots,P^{d+1*})$ with minimum cost over all such 
$(d+1)$-tuples.

In the {\sc Minimum Intertwined Cost Path} problem, the dependence of $\text{cost}(P^1;P^2,\ldots,P^{d+1})$ on the paths $P^2,\ldots,P^{d+1}$ is linear in the edges of $P^2,\ldots,P^{d+1}$.
Note that the IP problem from Sect.~\ref{sec:modeldesc} is a specific case of the {\sc Minimum Intertwined Cost Path} problem where 
the barrier acts as $P^1$, each of the dikes is represented one path $P^j$, $j=2,\ldots,d+1$,  and the cost functions $m_i^t$ only depends on the vertices $m_i^t(v^t_i,a_i^1)=m_i^t(v^t_i,v_i^1)$, in addition to the edges between $V_i$ and $V_{i+1}$ being restricted (only connecting vertices of non-decreasing heights).


This particular fact allowed us in Sect.~\ref{sec:CPB one dike one barrier} to give an algorithm for the problem, which runs in polynomial time if we consider the size of the sets in the partition of the vertices of the second graph as a constant.
Clearly if the bipartite graphs between $V^{(2)}_i$ and $V^{(2)}_{i+1}$ are complete, then this dynamic programming approach will not work.
It would be interesting to find out if some other approach may yield an efficient algorithm.

We end this section with some concrete questions.
\begin{que}
\label{que1}
  Is the {\sc Minimum Intertwined Cost Path} problem $\mathsf{NP}$-hard for unbounded number of possible heights?
\end{que}
We do not have an answer for Question~\ref{que1}, but we remark the following: with an appropriate cost function on the updating of the heights of one dike, instances of the {\sc Knapsack} problem can be seen as optimizing the height of one dike.
Indeed, the decision of updating the height of a dike at time $t\in \mathbb{N}$ corresponds to the decision of adding an certain number of copies of an item to the knapsack; the total height of the dike at time $t$ corresponds to the accumulated weight of the chosen items (counting multiplicities) to be carried among the first $t$ items.
The cost function of upgrading the height at time $t$ by $kw_t$ units corresponds to the profit of adding~$k$ copies of the item $t$, whose weight is $w_t$.
The cost function of the upgrading the dikes is such that once the capacity of the knapsack is exceeded by a set of items, then the cost of keeping or upgrading the dike height is unreasonable high.
With this correspondence, we observe that the optimal solution of the Knapsack problem corresponds to the optimal solution of the dike height.
Computing an optimal solution to the {\sc Knapsack} problem is well-known to be $\mathsf{NP}$-hard.
One of the inputs of the {\sc Knapsack} problem is the logarithm of the total weight of the knapsack bag.
Thus the dynamic program proposed in Sec.~\ref{sec:DP} is an exponential time algorithm.

If Question~\ref{que1} has a positive answer, then it makes sense to consider the following questions.
\begin{que}
  Under which conditions on the bipartite graphs $G^j[V^{(j)}_i,V^{(j)}_{i+1}]$, ($j=1,2$, $i=1,\ldots,k)$ is there a polynomial time algorithm for the {\sc Minimum Intertwined Cost Path} problem?
\end{que}

\begin{que}
  Suppose the partition size of $G_2$ is constant.
  Under which conditions on the bipartite graphs $G_j[V^{(j)}_i,V^{(j)}_{i+1}]$ ($j=1,2$, $i=0,\ldots,k)$ is there a polynomial time algorithm for the {\sc Minimum Intertwined Cost Path} problem?
\end{que}

\medskip
\noindent
\textbf{Acknowledgements.}
We thank Kees Roos for helpful discussions and for presenting the ideas of Brekelmans et al.~\cite{BHRE2012}.
We moreover thank Andr\'e Woning from Rijkswaterstaat for useful background information, and Stan van Hoesel and Peter van de Ven for their careful reading of the manuscript.
A large part of the results in this paper were obtained in the context of the 2017 \emph{Study Group Mathematics with Industry} week, for which we want to thank the organizers.



\end{document}